\documentclass[12pt]{article}
\usepackage{amsmath}
\usepackage{amsfonts}
\usepackage{amssymb}
\usepackage{amsthm}
\usepackage{graphicx,psfrag,epsf}
\usepackage{enumerate}
\usepackage{natbib}
\usepackage{hyperref}
\usepackage{url}
\usepackage[title]{appendix}

\bibpunct{(}{)}{;}{a}{}{,}

\newtheorem{prop}{Proposition}
\newtheorem{cor}{Corollary}

\begin{document}

\def\spacingset#1{\renewcommand{\baselinestretch}%
{#1}\small\normalsize} \spacingset{1}

%%%%%%%%%%%%%%%%%%%%%%%%%%%%%%%%%%%%%%%%%%%%%%%%%%%%%%%%%%%%%%%%%%%%%%%%%%%%%%

  \title{\bf Estimating Real Log Canonical Thresholds}
  \author{Toru Imai\footnote{imai.toru.7w@kyoto-u.ac.jp} \\
    Kyoto University}
  \date{}
  \maketitle

\bigskip
\begin{abstract}
	Evaluation of the marginal likelihood plays an important role in model selection problems. 
The widely applicable Bayesian information criterion (WBIC) and singular Bayesian information criterion (sBIC) give approximations to the log marginal likelihood, which can be applied to both regular and singular models. 
When the real log canonical thresholds are known, the performance of sBIC is considered to be better than that of WBIC, but only few real log canonical thresholds are known.
	In this paper, we propose a new estimator of the real log canonical thresholds based on the variance of thermodynamic integration with an inverse temperature.
	In addition, we propose an application to make sBIC widely applicable.
	Finally, we investigate the performance of the estimator and model selection by simulation studies and application to real data.
\end{abstract}

\noindent%
\vfill

%\newpage
\spacingset{1.45} % DON'T change the spacing!

\section{Introduction}
\label{sec:intro}
Let ${\bf X}^n=({\bf X}_1,...,{\bf X}_n)$ denote a sample of $n$ independent and identically distributed observations with each ${\bf X}_i \in \mathbb{R}^h$ drawn from a data generating distribution $q$. Let $M$ be a $d$-dimensional model with associated parameter ${\bf \theta} \in {\bf \Omega} \subset \mathbb{R}^{d}$, where ${\bf \Omega}$ is a parameter space. Let $p({\bf X}^n|{\bf \theta}, M)$ be the likelihood function and $\varphi({\bf \theta} | M)$ be a prior distribution. In this paper, we assume $p({\bf X}^n|{\bf \theta}, M)$ is differentiable and its first derivative function is not a constant.
The log marginal likelihood $\log \{ L(M) \}$ for model $M$ is defined as
\begin{equation*}
\log \{ L(M) \} := \log \int_{{\bf \Omega}} p({\bf X}^n | {\bf \theta}, M) \varphi({\bf \theta} | M) d{\bf \theta}.
\end{equation*}

A statistical model is termed regular if the mapping from a model parameter to a probability distribution is one-to-one and if its Fisher information matrix is positive definite. Otherwise, a statistical model is called singular.

There are numerous methods to approximate the marginal likelihood directly by using numerical integration approaches, for instance the Monte Carlo method (see, for example, \cite{Gelman1998} and \cite{Friel2012}), but according to \cite{Oates2016}, the estimators of marginal likelihood based on Monte Carlo sampling generally have high variance.

In the case of regular models, the Bayesian information criterion (BIC) \citep{Schwarz1978} gives an approximation of the log marginal likelihood up to $O_p(1)$. However, BIC is not applicable to singular models. The widely applicable Bayesian information criterion (WBIC) \citep{Watanabe2013} is proposed to approximate the log marginal likelihood up to $O_p[\surd \{\log (n)\}]$ for singular models and $O_p(1)$ for regular models. In addition, the Bayesian information criterion for singular models (sBIC) \citep{Drton2017b} approximates the log marginal likelihood up to $O_p(1)$ for models whose real log canonical thresholds and their multiplicities are known.

Another approach to evaluate the marginal likelihood is power posterior \citep{Friel2008}. Both WBIC and power posterior are based on thermodynamic integration. Although WBIC uses one temperature, power posterior uses many temperatures. Therefore, the computational cost of power posterior is much higher than that of WBIC \citep{Friel2017}.

When the real log canonical thresholds are known, simulation studies by \cite{Drton2017b} indicate that the performance of sBIC can be better than that of WBIC, but few real log canonical thresholds are known.

In this paper, we propose a new estimator of the real log canonical threshold based on the variance of thermodynamic integration with an inverse temperature. In addition, we propose an application to make sBIC widely applicable.

The paper is organized as follows: the key ideas and results are introduced in Section 2. 
We derive an estimator of the real log canonical threshold in Section 3. 
Then, we propose a widely applicable sBIC (WsBIC) in Section 4. 
We conduct numerical experiments to investigate the estimators of the real log canonical threshold and the performance of WsBIC in Section 5.  
Finally, we present some discussions in Section 6.

\section{Thermodynamic integration and WBIC}
\label{sec:wbic}

For any integrable function $f({\bf \theta})$ and a non-negative variable $t \in \mathbb{R}_{\ge 0}$, let $\mathbb{E}_{\bf \theta}^t \{f({\bf \theta}) \}$ and $\mathbb{V}_{\bf \theta}^t \{ f({\bf \theta}) \}$ be defined as
\begin{eqnarray*}
\mathbb{E}_{\bf \theta}^t \{ f({\bf \theta}) \} &:=& \frac{\displaystyle \int_{{\bf \Omega}} f({\bf \theta}) p({\bf X}^n | {\bf \theta}, M)^t \varphi ({\bf \theta} | M) d{\bf \theta}}{\displaystyle \int_{{\bf \Omega}} p({\bf X}^n | {\bf \theta}, M)^t \varphi ({\bf \theta} | M) d{\bf \theta}},  \\
\mathbb{V}_{\bf \theta}^t \{f({\bf \theta}) \} &:=& \mathbb{E}_{\bf \theta}^t \{f({\bf \theta})^2 \} - \Bigl[ \mathbb{E}_{\bf \theta}^t \{f({\bf \theta})\} \Bigr]^2,
\end{eqnarray*}
respectively. Here, $t$ is called an inverse temperature.

Let $F(t)$ be defined as
\begin{equation*}
F(t) := \log \int_{{\bf \Omega}} p({\bf X}^n | {\bf \theta}, M)^t \varphi({\bf \theta} | M) d{\bf \theta}.
\end{equation*}
Here, $F(0)=0, F(1)=\log L(M)$ by definition. Then, we obtain
\begin{equation*}
\log L(M) = \int_0^1 \frac{d}{dt} F(t) dt = \int_0^1 \mathbb{E}_{\bf \theta}^t \{ \log p({\bf X}^n | {\bf \theta}, M) \} dt. 
\end{equation*}

Since we have $d^2 F(t)/dt^2 = \mathbb{V}_{\bf \theta}^t \{ \log p({\bf X}^n | {\bf \theta}, M) \} > 0$, $d F(t)/dt = \mathbb{E}_{\bf \theta}^t \{ \log p({\bf X}^n | {\bf \theta}, M) \}$ is an increasing function. Hence, by the mean value theorem, there exists a unique temperature $t^* \in (0,1)$ such that
\begin{equation}
\log \{ L(M) \} = \mathbb{E}_{\bf \theta}^{t^*} \{ \log p({\bf X}^n | {\bf \theta}, M) \}. \label{unitemp}
\end{equation}

Based on equation (\ref{unitemp}), WBIC \citep{Watanabe2013} is defined as
\begin{equation*}
\mathrm{WBIC}(M) = \mathbb{E}_{\bf \theta}^{t_w} \{ \log p({\bf X}^n | {\bf \theta}, M) \}, 
\end{equation*}
where $t_w=1/\log (n)$.

Under a mild assumption that is given in the Appendix, \cite{Watanabe2000, Watanabe2001a, Watanabe2009b} showed that
\begin{equation*}
\log \{ L(M) \} = \log p({\bf X}^n | {\bf \theta}_0, M) - \lambda \log n + (\mathfrak{m}-1) \log \log (n) + O_p(1),
\end{equation*}
where $ {\bf \theta}_0$ is the parameter that minimizes the Kullback-Leibler divergence from a data generating distribution to a statistical model, $\lambda$ and $\mathfrak{m}$ are termed the real log canonical threshold and its multiplicity, respectively.

\cite{Watanabe2013} also showed that 
\begin{equation}
\mathbb{E}_{\bf \theta}^{t} [ \log p({\bf X}^n | {\bf \theta}, M) ]  = \log p({\bf X}^n | {\bf \theta}_0, M) - t^{-1}\lambda + t^{-1/2}V_n + O_p(1),  \label{thermo_asymp}
\end{equation}
where $t$ is a variable satisfying $t=c/\log n$ for some $c \in \mathbb{R}_{> 0}$ and $V_n$ is a random variable. In addition, the expectation of $V_n$ is equal to 0 and $V_n$ converges to $N(0, v_M)$ in law as $n \to \infty$, where $v_M \in \mathbb{R}_{> 0}$ is a constant.

Therefore, WBIC has the following properties:
\begin{eqnarray*}
\mathrm{WBIC}(M) &=& \log \{ L(M) \}+ O_p[\surd \{ \log (n) \}], \\
\mathbb{E}_{{\bf X}^n} \{ \mathrm{WBIC}(M) \} &=& \mathbb{E}_{{\bf X}^n} [\log \{L(M) \}] + O\{\log \log (n) \}.
\end{eqnarray*}

\section{Estimator of the real log canonical threshold}
\label{sec:est}

In statistics, the real log canonical threshold, also known as the {\it learning coefficient}, was first introduced by \cite{Watanabe1999, Watanabe2000, Watanabe2001a,Watanabe2009b}. The negative real log canonical threshold $(-\lambda)$ and its multiplicity $\mathfrak{m}$ are defined as the largest pole and its order of the zeta function, respectively:
\begin{equation*}
\zeta(z) := \int_\Omega K(\theta)^z \varphi(\theta) d\theta, 
\end{equation*}
where $K(\theta) = \int q(x) \log q(x)/p(x|\theta, M) dx$ and $z \in \mathbb{C}$.
Determining real log canonical thresholds is considered to be a challenging task. For work on real log canonical thresholds, see \cite{Aoyagi2009, Aoyagi2010a, Aoyagi2010b, Aoyagi2019},  \cite{Aoyagi2005}, \cite{Drton2017b}, \cite{Drton2017a}, \cite{Hayashi2017a, Hayashi2017b}, \cite{Rusakov2005}, \cite{Watanabe2001b}, \cite{Watanabe2003}, \cite{Watanabe2007}, \\
\cite{Yamazaki2003, Yamazaki2004, Yamazaki2005}, and \cite{Zwiernik2011}.

In the following sections, we introduce two consistent estimators of the real log canonical threshold.

\subsection{Estimator $\widehat{\lambda_\mathbb{E}}$}
\label{sec:est_e}

We consider estimating real log canonical thresholds by simulation. To distinguish a sample size $n$ of any given data, let $n_s$ denote the sample size of generated data for simulation.

Based on equation (\ref{thermo_asymp}), \cite{Watanabe2013} proposed an estimator $\widehat{\lambda_\mathbb{E}}$ of $\lambda$
\begin{eqnarray}
\widehat{\lambda_\mathbb{E}} &:=& - \frac{\mathbb{E}_{\bf \theta}^{t+\Delta} \{ \log p({\bf X}^{n_s} | {\bf \theta}, M) \} - \mathbb{E}_{\bf \theta}^{t} \{ \log p({\bf X}^{n_s} | {\bf \theta}, M) \}}{(t+\Delta)^{-1} - t^{-1}} \nonumber \\
&=& t \Bigl(t + \Delta \Bigr) \frac{\mathbb{E}_{\bf \theta}^{t+\Delta} \{ \log p({\bf X}^{n_s} | {\bf \theta}, M) \} - \mathbb{E}_{\bf \theta}^{t} \{ \log p({\bf X}^{n_s} | {\bf \theta}, M) \}}{\Delta}, \label{lambda_e}
\end{eqnarray}
where $\Delta=d/\log n_s \in \mathbb{R}_{> 0}$ and $t = c / \log n_s$ for some $c,d \in \mathbb{R}_{>0}$.

From equation (\ref{thermo_asymp}), $\widehat{\lambda_\mathbb{E}}$ has the following property:
\begin{equation*}
\widehat{\lambda_\mathbb{E}} = \lambda + O_p \left[ 1/\surd \{\log (n_s)\} \right].
\end{equation*}

Hence, $\widehat{\lambda_\mathbb{E}}$ is a consistent estimator of $\lambda$. 

We note that in principle, computing $\widehat{\lambda_\mathbb{E}}$ requires the Markov chain Monte Carlo method (MCMC) for each $\mathbb{E}_{\bf \theta}^{t}$ and $\mathbb{E}_{\bf \theta}^{t+\Delta}$. In order to reduce the computational cost, \cite{Watanabe2013} proposed an approximation $\widetilde{\mathbb{E}_{\bf \theta}^{t+\Delta}}$ of $\mathbb{E}_{\bf \theta}^{t+\Delta}$ by using $\mathbb{E}_{\bf \theta}^{t}$, which requires only one MCMC for $\mathbb{E}_{\bf \theta}^{t}$ to compute an approximation of the estimator $\widehat{\lambda_\mathbb{E}}$:
\begin{equation}
\widetilde{\mathbb{E}_{\bf \theta}^{t+\Delta}} \{ \log p({\bf X}^{n_s} | {\bf \theta}, M) \} := \frac{\mathbb{E}_{\bf \theta}^{t} [ \log p({\bf X}^{n_s} | {\bf \theta}, M) \exp\{\Delta \log p({\bf X}^{n_s} | {\bf \theta}, M) \} ] }{\mathbb{E}_{\bf \theta}^{t} [ \exp\{\Delta \log p({\bf X}^{n_s} | {\bf \theta}, M)\} ]}. \label{E_approx}
\end{equation}

Let $\widetilde{\lambda_\mathbb{E}}$ be defined as
\begin{equation*}
\widetilde{\lambda_\mathbb{E}} := t \Bigl(t + \Delta \Bigr) \frac{ \widetilde{\mathbb{E}_{\bf \theta}^{t+\Delta}} \{ \log p({\bf X}^{n_s} | {\bf \theta}, M) ] - \mathbb{E}_{\bf \theta}^{t} [ \log p({\bf X}^{n_s} | {\bf \theta}, M) \}}{\Delta}.
\end{equation*}

We demonstrate the difference between $\widehat{\lambda_\mathbb{E}}$ and $\widetilde{\lambda_\mathbb{E}}$ in simulation studies in Section \ref{sec:exp_gmm}.

In the next section, we propose a new consistent estimator.

\subsection{Estimator $\widehat{\lambda^m_\mathbb{V}}$}
\label{sec:est_v}

Let $\widehat{\lambda^1_{\mathbb{V}}}({\bf X}^n)$ be defined as
\begin{equation*}
\widehat{\lambda^1_{\mathbb{V}}}({\bf X}^n) := t^2 \mathbb{V}_{\bf \theta}^t \{ \log p({\bf X}^n | {\bf \theta}, M) \},
\end{equation*}
where $t = c / \log n_s$ for some $c \in \mathbb{R}_{> 0}$.

\begin{prop} \label{prop_1}
 $\widehat{\lambda^1_{\mathbb{V}}}({\bf X}^n)$ is a consistent estimator of $\lambda$.
\end{prop}

\begin{proof}

Taking the derivative of equation (\ref{thermo_asymp}) gives
\begin{equation}
\mathbb{V}_{\bf \theta}^t [ \log p({\bf X}^n | {\bf \theta}, M) ] = t^{-2}\lambda - \frac{t^{-3/2}}{2} V_n + \xi_n, \label{proof1}
\end{equation}
where $\xi_n$ is a random variable such that $\xi_n=o_p[\{ \log (n) \}^{-3/2}]$.
Multiplying both sides of the equation (\ref{proof1}) by $t^2$, we obtain
\begin{equation}
\widehat{\lambda^1_{\mathbb{V}}}({\bf X}^n) = \lambda - \frac{t^{1/2}}{2} V_n  + t^2 \xi_n  = \lambda + O_p \left[ 1/ \surd \{\log (n) \} \right]. \label{lambda_hat_asymp}
\end{equation}

Therefore, the estimator $\widehat{\lambda^1_{\mathbb{V}}}({\bf X}^n)$ is consistent.
\end{proof}

Let $X^{(n_s, k)}$ be a sample of $n_s$ independent and identically distributed observations for each $k=1,2,...,m$, and let $\widehat{\lambda^m_{\mathbb{V}}}$ be defined as
\begin{equation*}
\widehat{\lambda^m_{\mathbb{V}}} := \frac{\sum_{k=1}^{m} \widehat{\lambda^1_{\mathbb{V}}}({\bf X}^{(n_s,k)})}{m}.
\end{equation*}
Then, as a corollary of Proposition \ref{prop_1}, we obtain the following:

\begin{cor} \label{cor_1}
$\widehat{\lambda^m_{\mathbb{V}}}$ is a consistent estimator of $\lambda$.
\end{cor}

\begin{prop} \label{prop_2}
$\widehat{\lambda^m_{\mathbb{V}}}$ is asymptotically normal:
\begin{equation*}
\widehat{\lambda^m_{\mathbb{V}}} - \lambda \xrightarrow{d} N(0, \sigma^2_{\lambda}),
\end{equation*}
where $\displaystyle \sigma^2_{\lambda} = \frac{cv_M}{4m \log n_s}$ and $c \in \mathbb{R}_{>0}$.
\end{prop}

\begin{proof}
Let $V_{(n,k)}$ denote $V_n$ for $X^{(n, k)}$. From equation (\ref{lambda_hat_asymp}), we have
\begin{equation}
\widehat{\lambda^m_{\mathbb{V}}} = \lambda - \frac{1}{2m} \left(\sum_{k=1}^{m} V_{(n,k)} \right) \surd \{c/\log (n_s)\}  + o_p \left[ 1/\surd \{\log (n_s) \} \right].
\end{equation}
Since $V_{(n,k)} \xrightarrow{d} N(0, v_M)$ for each $k$ \citep{Watanabe2013}, we obtain the result.
\end{proof}

Note that according to Proposition \ref{prop_2}, we can reduce the bias by increasing $n_s$ and reduce the variance by increasing $n_s$ and $m$.
It is important to reduce not only bias but also variance, because for nested models $M_{i'} \subset M_{i}$ in model selection, the value of $\lambda (i',j)$ should be less than that of $\lambda (i,j)$, but their estimates may be reversed due to variance.
From Proposition \ref{prop_2}, increasing $m$ can reduce the variance in the order $1/m$, whereas increasing $n_s$ can reduce the variance in the order $1/\log n_s$.
In addition, we can compute $\widehat{\lambda^1_{\mathbb{V}}}({\bf X}^{(n_s,k)})$ in parallel. Therefore, the computation time of $\widehat{\lambda^m_{\mathbb{V}}}$ can be the same as that of $\widehat{\lambda^1_{\mathbb{V}}}$.

\subsection{Optimal choice of the hyperparameter}
\label{sec:opt_c}

As we have seen in the previous sections, $\widehat{\lambda_{\mathbb{E}}}$ has hyperparameters $(c, d)$ and $\widehat{\lambda^1_{\mathbb{V}}}$ has the hyperparameter $c$.
First, let us try to find an optimal choice of $c$ and $d$ that makes $\widehat{\lambda_{\mathbb{E}}}$ an unbiased estimator.
As in the same argument in equation (\ref{unitemp}), we have the unique temperature $t^*$ such that
\begin{equation*}
\mathbb{E}_{\bf \theta}^{t^*} \{ \log p({\bf X}^{n_s} | {\bf \theta}, M) \} = \log p({\bf X}^{n_s} | {\bf \theta_0}, M) + \frac{\lambda}{t^*}. 
\end{equation*}

Then, we have
\begin{equation*}
\mathbb{E}_{\bf \theta}^{t^* + \delta} \{ \log p({\bf X}^{n_s} | {\bf \theta}, M) \} = \log p({\bf X}^{n_s} | {\bf \theta_0}, M) + \frac{\lambda}{t^* + \delta} + \eta_n,
\end{equation*}
where $\eta_n$ is a random variable and $\delta \in \mathbb{R}_{>0}$. $\eta_n$ takes a positive value because $\mathbb{E}_{\bf \theta}^{t} \{ \log p({\bf X}^{n_s} | {\bf \theta}, M) \}$ is an increasing function,
\begin{eqnarray*}
\mathbb{E}_{\bf \theta}^{t^*} \{ \log p({\bf X}^{n_s} | {\bf \theta}, M) \} &<& \mathbb{E}_{\bf \theta}^{t^* + \delta} \{ \log p({\bf X}^{n_s} | {\bf \theta}, M) \} \nonumber \\
\Rightarrow \log p({\bf X}^{n_s} | {\bf \theta_0}, M) + \frac{\lambda}{t^*} &<& \log p({\bf X}^{n_s} | {\bf \theta_0}, M) + \frac{\lambda}{t^* + \delta} + \eta_n \nonumber \\
 &<& \log p({\bf X}^{n_s} | {\bf \theta_0}, M) + \frac{\lambda}{t^*} + \eta_n \nonumber \\
\Rightarrow 0 &<& \eta_n.
\end{eqnarray*}
In addition, since the optimal temperature is unique, $\eta_n \neq 0$ unless $d=0$ and as $d \to 0$, $\eta_n \to 0$. Therefore, the only candidate optimal point $(c^*, d^*)$ for $\widehat{\lambda_{\mathbb{E}}}$ in the neighborhood of $t^*$ is $(c^*, d^*)=(t^*\log (n), 0)$. However, when $d=0$, we have $\Delta=0$ and then $\widehat{\lambda_{\mathbb{E}}}$ is incomputable. Therefore, there is no optimal point $(c^*, d^*)$ for $\widehat{\lambda_{\mathbb{E}}}$ in the neighborhood of $t^*$. In Section \ref{sec:exp_gmm}, we perform numerical experiments and demonstrate the behavior of $\widehat{\lambda_{\mathbb{E}}}$ when $d$ is small. In addition, since $\mathbb{E}_{\bf \theta}^{t} \{ \log p({\bf X}^{n_s} | {\bf \theta}, M) \}$ is an increasing function, the above argument holds for any $t$. Hence, we can not construct an unbiased estimator based on $\widehat{\lambda_{\mathbb{E}}}$.

Next, let us find an optimal choice of $c$ such that $\widehat{\lambda^1_{\mathbb{V}}}$ is an unbiased estimator. Taking the limit of $\Delta$ in equation (\ref{lambda_e}) gives:
\begin{equation*}
\widehat{\lambda^1_{\mathbb{V}}} = \lim_{\Delta \to 0} \widehat{\lambda_{\mathbb{E}}}.
\end{equation*}

Therefore, an optimal point of $c$ for $\widehat{\lambda^1_{\mathbb{V}}}$ is
\begin{equation}
c^* = t^* \log (n). \label{opt_c}
\end{equation}

On the other hand, \cite{Watanabe2013} showed that
\begin{equation}
t^* = \frac{1}{\log (n)} + o_p\left\{ \frac{1}{\log (n)} \right\}. \label{opt_t}
\end{equation}
Here, the term $o_p\{ 1/\log (n) \}$ depends on a model, a prior, and a data-generating distribution. Therefore, the model-free term is the only leading term. Hence, from equations (\ref{opt_c}) and (\ref{opt_t}), the optimal point $c^*$ for $\widehat{\lambda^1_{\mathbb{V}}}$ is
\begin{equation*}
c^* = 1 + o_p(1),
\end{equation*}
where the $o_p(1)$ depends on a model, a prior, and a data-generating distribution.

\subsection{Effective number of parameters}
\label{sec:effnum}

Here we compare the real log canonical threshold with the effective number of parameters.

The effective number of parameters was introduced by \cite{Spiegelhalter2002} as 
\begin{equation*}
p_D = - 2\mathbb{E}_{\bf \theta}^{t=1} \{ \log p({\bf X}^{n} | {\bf \theta}, M) \} + 2\log p\{ {\bf X}^{n} |  \mathbb{E}_{\bf \theta}^{t=1} ({\bf \theta}), M) \}.
\end{equation*}

$p_D$ is not invariant under reparametrization and can have negative values \citep{Spiegelhalter2014}. \cite{Gelman2004} proposed a modified effective number of parameters:
\begin{equation*}
p_{\mathbb{V}} =  2 \mathbb{V}_{\bf \theta}^{t=1} \{ \log p({\bf X}^n | {\bf \theta}, M) \}.
\end{equation*}

$p_{\mathbb{V}}$ is invariant to reparametrization and is always positive \citep{Spiegelhalter2014}. We note that \cite{Watanabe2010} showed $p_{\mathbb{V}}/2$ is an asymptotically biased estimator of real log canonical thresholds in general:
\begin{equation}
\lim_{n \to \infty} \mathbb{E}_{{\bf X}^n} (p_{\mathbb{V}}/2) = \lambda + \nu'(1), \label{pV_asymp}
\end{equation}
where $\nu(t)$ is called the singular fluctuation and is defined as
\begin{equation*}
\nu(t) := \lim_{n \to \infty} \frac{t}{2} \mathbb{E}_{{\bf X}^n} \Bigl[ \sum_{i=1}^n  \mathbb{V}_{\bf \theta}^{t} \{ \log p({\bf X}_i | {\bf \theta}, M) \}  \Bigr],
\end{equation*}
and $\nu'(t)$ is the first derivative of $\nu(t)$.

We investigate the performances of $p_{\mathbb{V}}/2$ and the other estimators of the real log canonical threshold in numerical experiments in Section \ref{sec:exp_gmm}.

\section{Widely applicable sBIC}
\label{sec:wsbic}

In this section, we first briefly introduce sBIC \citep{Drton2017b} and then propose the widely applicable sBIC.

Let $I$ be a finite index, $\{ M_i | i \in I \}$ a set of candidate models, $p(M_i)$ a prior probability of model $M_i$, and $p(M_i|{\bf X}^n)$ its posterior probability of model $M_i$. We define $i \preceq j$ for $i, j \in I$ when $M_i \subseteq M_j$.
Let $\lambda(i,j)$ and $\mathfrak{m}(i,j)$ be the real log canonical threshold and its multiplicity of $M_i$ with the data-generating distribution $q \in M_j$. 
$L_{ij}$ is defined as
\begin{equation*}
L_{ij} := p({\bf X}^n | \hat{\bf \theta}_i, M_i) \frac{(\log n)^{\mathfrak{m}(i,j) - 1}}{n^{\lambda(i,j)}},
\end{equation*}
where $\hat{\bf \theta}_i$ is the maximum likelihood estimator of $\theta_i$.

sBIC for model $M_i$ is based on a weighted average of $L_{ij}$ by posterior probabilities $p(M_j|{\bf X}^n)$:
\begin{equation}
S(M_i) := \frac{\sum_{j \preceq i} L_{ij} p(M_j|{\bf X}^n)}{\sum_{j \preceq i} p(M_j|{\bf X}^n)} = \frac{\sum_{j \preceq i} L_{ij} p(M_j) L(M_j) }{\sum_{j \preceq i} p(M_j) L(M_j)}. \label{sbic_pre}
\end{equation}

Here, the marginal likelihood $L(M_j)$ is what we would like to evaluate, and then by replacing $L(M_j)$ by $S(M_j)$ in equation (\ref{sbic_pre}), we obtain
\begin{equation}
S(M_i) = \frac{\sum_{j \preceq i} L_{ij} p(M_j) S(M_j) }{\sum_{j \preceq i} p(M_j) S(M_j)}. \label{sbic}
\end{equation}

\cite{Drton2017b} showed that equation (\ref{sbic}) has the unique positive solution $S(M_i)_{+}$ and defined sBIC for model $M_i$ as
\begin{equation*}
\mathrm{sBIC}(M_i) := \log S(M_i)_{+}.
\end{equation*}
$\mathrm{sBIC}(M_i)$ has the following asymptotic property:
\begin{equation*}
\mathrm{sBIC}(M_i) = \log L(M_i) + O_p(1).
\end{equation*}

For cases in which $\lambda(i,j)$ and $\mathfrak{m}(i,j)$ are unknown, \cite{Drton2017b} proposed using an upper bound $\overline{\lambda(i,j)}$ of $\lambda(i,j)$ and the lower bound of $\mathfrak{m}(i,j),$ the latter of which is equal to 1.
sBIC in this manner is denoted by $\overline{\rm sBIC}$.

However, in general, it is difficult to accurately compute $L_{ij}$ since few exact real log canonical thresholds or tight upper bounds are known.
Therefore, instead of $L_{ij}$, we propose to use $\widehat{L_{ij}}$:
\begin{equation*}
\widehat{L_{ij}} := p({\bf X}^n | \hat{\bf \theta}_i, M_i) \frac{1}{n^{\widehat{\lambda^m_\mathbb{V}}(i,j)}}.
\end{equation*}
To obtain $\widehat{\lambda^m_\mathbb{V}}(i,j)$, first we set $n_s, m \in \mathbb{Z}$, then generate data ${\bf X}^{(n_s,k)}=({\bf X}_1,...,{\bf X}_{n_s})_k$ independently from $M_j$ for each $k=1,2,...,m$, and finally compute:
\begin{equation*}
\widehat{\lambda^m_{\mathbb{V}}}(i,j) = \frac{\sum_{k=1}^m \widehat{\lambda^1_{\mathbb{V}}}({\bf X}^{(n_s,k)})}{m}.
\end{equation*}
We call sBIC based on $\widehat{L_{ij}}$ the widely applicable sBIC (WsBIC).

\section{Numerical experiments}
\label{sec:exp}

We conduct three numerical experiments and one application to real data to investigate the performances of $\widehat{\lambda^m_\mathbb{V}}$ and WsBIC.
First, we compare  the estimators $\widehat{\lambda^m_\mathbb{V}}$, $\widetilde{\lambda^m_\mathbb{V}}$, $\widehat{\lambda_\mathbb{E}}$, and $p_\mathbb{V}/2$ to assess the biases and variances of the estimators.
Second, we evaluate the bias of $\widehat{\lambda^m_\mathbb{V}}$ and compare the performances of sBIC and WsBIC for a model in which the exact real log canonical thresholds are known.
Third, we compare the performances of $\overline{\rm sBIC}$, WsBIC, and WBIC for a model in which only the upper bounds of the real log canonical thresholds are known.
Finally, we apply WsBIC to real data and compare BIC and WBIC.

\subsection{Comparison of the estimators of real log canonical thresholds}
\label{sec:exp_gmm}

In this section, we consider the following mixture model with two normal distributions:
\begin{equation*}
\alpha N(\mu_1, 1) + (1-\alpha) N(\mu_2, 1).
\end{equation*}
When the data generating model is $N(0, 1)$, \cite{Aoyagi2010a} showed the real log canonical threshold $\lambda$ is 3/4.

We conduct 1000 simulations for each sample size $n_s=50, 100, 200, 500, 1000$ to compute $\widehat{\lambda^m_\mathbb{V}}$, $\widehat{\lambda_\mathbb{E}}$, $\widetilde{\lambda_\mathbb{E}}$, and $p_\mathbb{V}/2$.
We set the prior $\alpha \sim {\rm Unif}(0,1)$, $\mu_1, \mu_2 \sim N(0, 4)$, $c=1$, and $d=1/10, 1, 10$ for $\widehat{\lambda_\mathbb{E}}$ and $\widetilde{\lambda_\mathbb{E}}$, and $m=1, 10, 100$ for $\widehat{\lambda^m_\mathbb{V}}$. 
We use the Hamiltonian Monte Carlo method, implemented in the R package {\it RStan} \citep{RStan}, to obtain the posteriors.

\begin{table}[htbp]
\begin{center}
\caption{Estimates of the real log canonical threshold of the Gaussian mixture model by each method}
{\fontsize{8pt}{8pt}\selectfont
\begin{tabular}{@{\extracolsep{4pt}}lcrrrrrrrrrr@{}} \hline
Method & \# of & \multicolumn{10}{c}{$n_s$}    \\ 
       &MCMC   &  \multicolumn{2}{c}{50}& \multicolumn{2}{c}{100}& \multicolumn{2}{c}{200}& \multicolumn{2}{c}{500}& \multicolumn{2}{c}{1000}\\ \cline{3-4} \cline{5-6} \cline{7-8} \cline{9-10} \cline{11-12}
       &   &  mean &(s.d.)& mean &(s.d.)& mean &(s.d.)& mean &(s.d.)& mean &(s.d.)\\ \hline
$\lambda$ (exact) &  - &  0.750 &  - &  0.750 &  - &  0.750 & - & 0.750 & - & 0.750 & - \\
$\widehat{\lambda^1_\mathbb{V}}$                    &  1 &  0.835 &  (0.138) & 0.817  &  (0.134)  & 0.807 & (0.143) & 0.787 & (0.132) & 0.764 & (0.139)\\
$p_\mathbb{V}/2$                                            &  1 &  0.866 &  (0.205) & 0.845   & (0.237)   & 0.837 & (0.286) & 0.813 & (0.286) & 0.786 & (0.272) \\
$\widetilde{\lambda_\mathbb{E}}_{(d=1/10)}$ &  1 &  1.544 &  (0.951) &  1.530 &  (0.826) & 1.523 & (0.847) & 1.479 & (0.776) & 1.430 & (0.833) \\ 
$\widetilde{\lambda_\mathbb{E}}_{(d=1)}$       &  1 &  1.148 &  (0.151) &  1.112 &  (0.130) & 1.092 & (0.146) & 1.054 & (0.130) & 1.026 & (0.131) \\ 
$\widetilde{\lambda_\mathbb{E}}_{(d=10)}$     &  1 &  0.927 &  (0.117) &  0.898 &  (0.108) & 0.880 & (0.121) & 0.847 & (0.105) & 0.821 & (0.098) \\
$\widehat{\lambda_\mathbb{E}}_{(d=1/10)}$   &  2 &  0.841 &  (1.023) &  0.839 &  (0.930) & 0.850 & (0.893) & 0.829 & (0.819) & 0.798 & (0.887) \\ 
$\widehat{\lambda_\mathbb{E}}_{(d=1)}$        &  2 &  0.863 &  (0.154) &  0.853 &  (0.154) & 0.824 & (0.148) & 0.799 & (0.129) & 0.777 & (0.137) \\ 
$\widehat{\lambda_\mathbb{E}}_{(d=10)}$      &  2 &  0.877 &  (0.109) &  0.856 &  (0.108) & 0.832 & (0.109) & 0.800 & (0.094) & 0.778 & (0.088) \\ \hline 
$\widehat{\lambda^{10}_\mathbb{V}}$                &  10 &  0.837 &  (0.044) &  0.820 &  (0.046) &  0.810 & (0.045) & 0.785 & (0.041) & 0.762 & (0.042)\\
$\widehat{\lambda^{100}_\mathbb{V}}$            &  100 &  0.835 &  (0.013) &  0.817 &  (0.014) &  0.806 & (0.013) & 0.787 & (0.012) & 0.763 & (0.013)\\ \hline
\end{tabular}
}
\label{tab_gmm}
\end{center}
\end{table}

The results of the simulations are shown in Table \ref{tab_gmm} and Figure \ref{fig:gmm}.

First, Table \ref{tab_gmm} shows that the means of all the estimates seem to approach the true value $3/4$ as $n_s$ increases. Among the methods that use MCMC once or twice, $\widehat{\lambda^1_\mathbb{V}}$ has the lowest bias of all $n_s$. 
The performances of $\widehat{\lambda_\mathbb{E}}$ and $\widetilde{\lambda_\mathbb{E}}$ depend on $d$. As $d$ increases, the variances of $\widehat{\lambda_\mathbb{E}}$ and $\widetilde{\lambda_\mathbb{E}}$ decrease and the bias of $\widetilde{\lambda_\mathbb{E}}$ decreases for each $n_s$. 
Comparing $\widehat{\lambda_\mathbb{E}}$ and $\widetilde{\lambda_\mathbb{E}}$, the estimate of $\widetilde{\lambda_\mathbb{E}}$ is larger than that of $\widehat{\lambda_\mathbb{E}}$ for each $d$ and $n_s$. This may due to the bias caused by approximating the thermodynamic integration in equation (\ref{E_approx}).
When $d=1/10$, both $\widehat{\lambda_\mathbb{E}}$ and $\widetilde{\lambda_\mathbb{E}}$ have high variance. In addition, $\widetilde{\lambda_\mathbb{E}}$ has high bias when $d=1/10, 1$. Regarding bias, $\widehat{\lambda_\mathbb{E}}$ with $d=1$ is the best among  $\widehat{\lambda_\mathbb{E}}$ and $\widetilde{\lambda_\mathbb{E}}$.
$p_\mathbb{V}/2$ has higher bias and variance than $\widehat{\lambda^1_\mathbb{V}}$. This might reflect the effect of $\nu'(1)$ as seen in equation (\ref{pV_asymp}).

Next, from Figure \ref{fig:gmm}, we see that both $\widehat{\lambda_\mathbb{E}}$ and $\widetilde{\lambda_\mathbb{E}}$ take negative values whereas $\widehat{\lambda^1_\mathbb{V}}$ and $p_\mathbb{V}/2$ do not. As $m$ increases, the variances of $\widehat{\lambda^m_\mathbb{V}}$ decrease. From Proposition \ref{prop_2}, increasing $n_s$ and $m$ decreases the variance of $\widehat{\lambda^m_\mathbb{V}}$, and increasing $m$ most efficiently decreases the variance since the asymptotic variance has factor $1/m$ for $m$, compared to $1/\log n_s$ for $n_s$.

Here, we note that the variance of the estimator is important because the smaller the variance, the smaller the probability of $\widehat{\lambda}(i',j) > \widehat{\lambda}(i,j)$ for nested models $M_{i'} \subset M_i$. In model selection, $\widehat{\lambda}(i',j) < \widehat{\lambda}(i,j)$ for nested models $M_{i'} \subset M_i$ is an essential condition and we need to recompute $\widehat{\lambda}(i',j), \widehat{\lambda}(i,j)$ if we obtain  $\widehat{\lambda}(i',j) > \widehat{\lambda}(i,j)$.

\begin{figure}
\includegraphics[width=\textwidth]{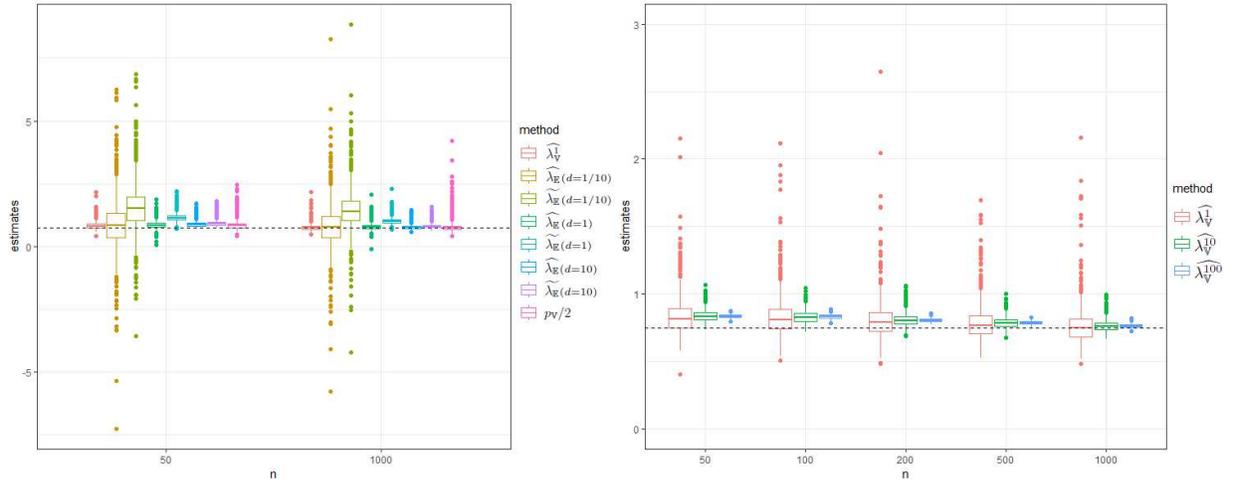}
\caption{ Boxplots of estimates of the real log canonical threshold of the Gaussian mixture model. The black dashed line indicates a value of $3/4$, which is the true real log canonical threshold. (left) The estimates of $\widehat{\lambda_\mathbb{E}}$, $\widetilde{\lambda_\mathbb{E}}$, $\widehat{\lambda^1_\mathbb{V}}$, and $p_\mathbb{V}/2$. (right) The estimates of $\widehat{\lambda^1_\mathbb{V}}$, $\widehat{\lambda^{10}_\mathbb{V}}$, and $\widehat{\lambda^{100}_\mathbb{V}}$.}
\label{fig:gmm}
\end{figure}

\subsection{Comparison of $\widehat{\lambda^m_\mathbb{V}}$ and $\lambda$ and performance of sBIC and WsBIC}

Here, we estimate the real log canonical thresholds and compare the exact values using reduced rank regression because this is the singular model in which the exact real log canonical thresholds are known \citep{Aoyagi2005}. 

Reduced rank regression is defined as follows \citep{Reinsel1998}. Let $Y \in \mathbb{R}^N, X \in \mathbb{R}^M,$ and rank $H \in \mathbb{Z}$ and $0 \le H \le \min \{ M, N \}$. Let the parameter $C \in \mathbb{R}^{N \times M}$ and ${\rm rank} (C) \le H$. Then, for sample size $n$, reduced rank regression is defined as
\begin{equation*}
Y \sim N(CX, I_n \otimes I_N).
\end{equation*}

To compute $\widehat{\lambda^m_\mathbb{V}}(i,j)$, we set $n_s=10000, m=100$ and use the Metropolis Hastings method as in the program code on Sumio Watanabe's website  (\\
\noindent \url{http://watanabe-www.math.dis.titech.ac.jp/users/swatanab/wbic_reduced.m}).

The results are summarized in Table \ref{tab_rrr}.
Clearly, the table shows that the estimates $\widehat{\lambda^m_\mathbb{V}}(i,j)$ are close to $\lambda (i,j)$. 
To obtain more precise values of $\lambda (i,j)$, it is sufficient to simply increase $n_s$.

When $n_s$ is large, sBIC and WsBIC have almost same values since $\widehat{\lambda^m_\mathbb{V}}(i,j)$ and $\lambda (i,j)$ have similar values. Therefore, the performances of sBIC and WsBIC are expected to be similar. To demonstrate this,  we conduct 200 simulations for each sample size $n=10, 20, 50$ to compute sBIC and WsBIC based on $\lambda (i,j)$ and $\widehat{\lambda^m_\mathbb{V}}(i,j)$ in Table \ref{tab_rrr}, implemented in the R package {\it sBIC} \citep{Weihs2016}. We also compute WBIC for comparison.

The results are shown in Figure \ref{fig_rrr}. The performances of sBIC and WsBIC are similar and outperform that of WBIC. Here, we would like to emphasize that WsBIC and WBIC do not require theoretical values of  $\lambda (i,j)$, whereas sBIC does. Hence, WsBIC can be widely applied and is expected to have the almost same performance as sBIC.

\begin{table*}[t]
\begin{center}
{\fontsize{9pt}{9pt}\selectfont
\caption{ $\lambda (i,j)$ and $\widehat{\lambda^m_{\mathbb{V}}}(i,j)$ for reduced rank regression with values of $i,j$ ($i \ge j$),  $(n_s,m)=(2000,100)$.}
\begin{tabular}{@{\extracolsep{4pt}}lrrrrrrrrrr@{}} \hline
 & \multicolumn{2}{c}{$j=1$} & \multicolumn{2}{c}{$j=2$} & \multicolumn{2}{c}{$j=3$} & \multicolumn{2}{c}{$j=4$} & \multicolumn{2}{c}{$j=5$} \\ 
 & $\lambda (i,j)$ & $\widehat{\lambda^m_{\mathbb{V}}}(i,j)$ & $\lambda (i,j)$ & $\widehat{\lambda^m_{\mathbb{V}}}(i,j)$ & $\lambda (i,j)$ & $\widehat{\lambda^m_{\mathbb{V}}}(i,j)$ & $\lambda (i,j)$ & $\widehat{\lambda^m_{\mathbb{V}}}(i,j)$ & $\lambda (i,j)$ & $\widehat{\lambda^m_{\mathbb{V}}}(i,j)$ \\ \cline{2-3} \cline{4-5} \cline{6-7} \cline{8-9} \cline{10-11} 
$i=1$ & 5.5  & 5.50  &  & & &   & & & &\\
$i=2$ & 8  & 7.91  & 10  & 10.01  & &  & & & &\\
$i=3$ & 10  & 9.92  & 12  & 11.75  & 13.5  & 13.49 & & & & \\
$i=4$ & 12  & 11.75  & 13.5  & 13.30  & 15  & 14.79 & 16 & 16.02 & & \\
$i=5$ & 13.5  & 13.32  & 15  & 14.65  & 16  & 15.81 & 17 & 16.79 & 17.5 & 17.35  \\
\hline
\end{tabular}
}
\label{tab_rrr}
\end{center}
\end{table*}

\begin{figure*}
\includegraphics[width=\textwidth]{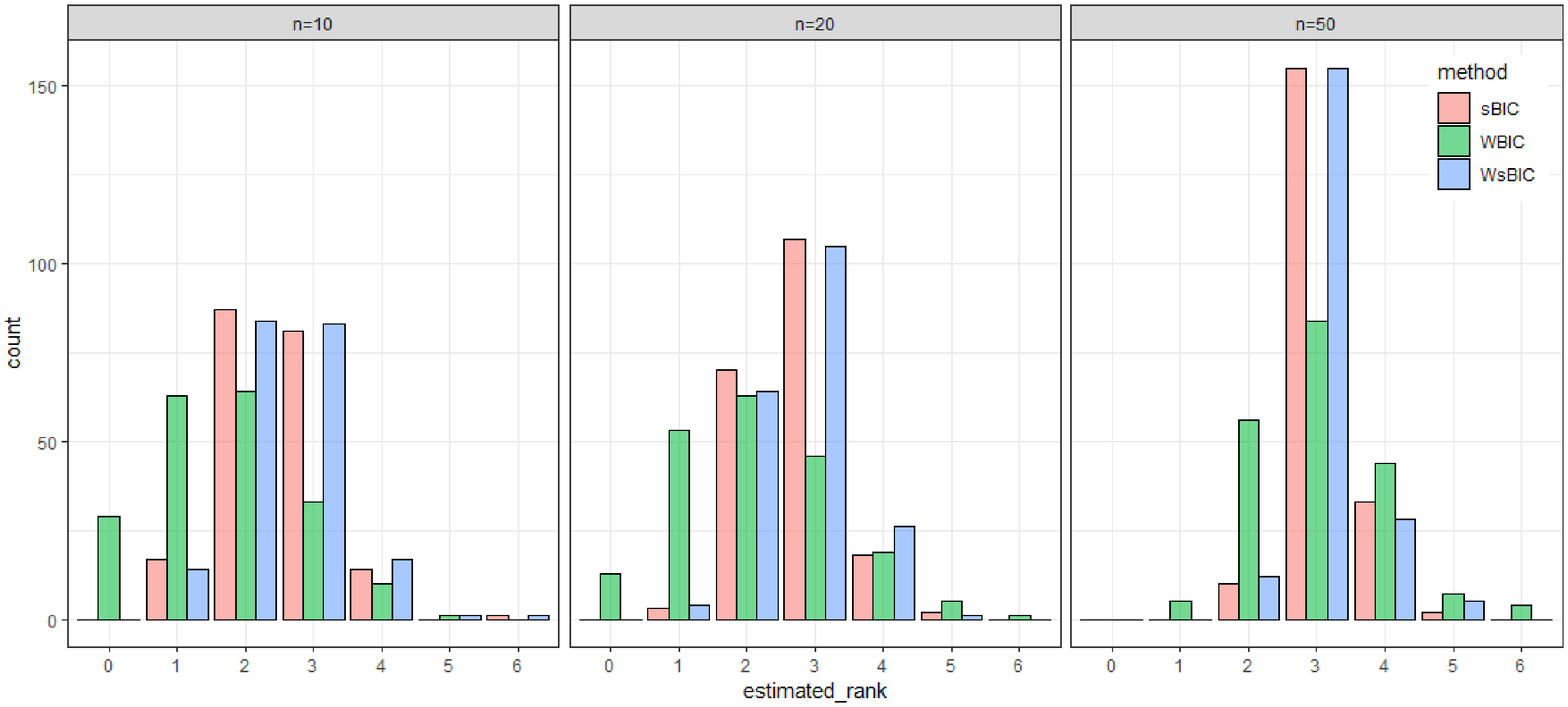}
\caption{Performances of BIC, sBIC, and WsBIC in the reduced regression model}
\label{fig_rrr}
\end{figure*}

\subsection{Comparison of $\widehat{\lambda^m_\mathbb{V}}$ and $\bar{\lambda}$ and performances of $\overline{\rm sBIC}$ and WsBIC}
\label{sec:bmm_exp}

Interesting cases in which to compare sBIC and WsBIC are the models whose upper bounds of the real log canonical thresholds are known, but the exact values are not.
An example is the binomial mixture model. With $i$ as the number of mixture components, the binomial mixture model is defined as
\begin{equation*}
\sum_{h=1}^i \pi_h B(k, p_h),
\end{equation*}
where $\pi_h \ge 0$, $\sum_{h=1}^i \pi_h=1$, and $B(k, p_h)$ is a binomial distribution.

\cite{Drton2017b} proposed two $\overline{\rm sBIC}$ based on two different upper bounds of the real log canonical thresholds:
\begin{eqnarray*}
\bar{\lambda}^{1}(i,j) &:=& \frac{i+j}{2} - \frac{1}{2}, \\
\bar{\lambda}^{0.5}(i,j) &:=& \frac{i+3j}{4} - \frac{1}{2}.
\end{eqnarray*}
sBIC based on $\bar{\lambda}^{1}(i,j)$ and $\bar{\lambda}^{0.5}(i,j)$ are denoted as $\overline{\rm sBIC}_{1}$ and $\overline{\rm sBIC}_{0.5}$, respectively.
Note that $\bar{\lambda}^{1}(i,j)$ is derived by simple parameter counting \citep{Watanabe2009b} whereas the derivation of $\bar{\lambda}^{0.5}(i,j)$ requires more complicated analysis \citep{Rousseau2011}.

To estimate $\widehat{\lambda^m_\mathbb{V}}(i,j)$, we set $n_s=10000, m=100$ and generate data from the binomial mixture model of $j$ components with sample size parameter $k=30$ and mixture weight $\pi_h = 1/j$ for each $h=1,...,j$. We set the prior of $\pi_h$ to the flat Dirichlet distribution and $p_h \sim {\rm logit}^{-1}({\rm Unif}(-\infty, \infty))$. We use the Hamiltonian Monte Carlo method to obtain the posteriors implemented in the R package {\it RStan} \citep{RStan}.

The results are shown in Table \ref{tab_bmm}.
In the above setting, most values of $\widehat{\lambda^m_\mathbb{V}}(i,j)$ lie between $\bar{\lambda}^{1}(i,j)$ and $\bar{\lambda}^{0.5}(i,j)$.

\begin{table*}[htbp]
\begin{center}
\caption{ $\bar{\lambda}^{1} (i,j)$, $\bar{\lambda}^{0.5} (i,j)$, and $\widehat{\lambda^m_{\mathbb{V}}}(i,j)$ for binomial mixture model with values of $i,j$ ($i \ge j$), $(n_s,m)=(10000,100)$}
\begin{tabular}{@{\extracolsep{4pt}}lrrrrrrrr@{}} \hline
 & \multicolumn{4}{c}{$j=1$} & \multicolumn{4}{c}{$j=2$}  \\ 
 & $d/2$ & $\bar{\lambda}^1 (i,j)$ & $\bar{\lambda}^{0.5} (i,j)$ & $\widehat{\lambda^m_{\mathbb{V}}}(i,j)$ & $d/2$ & $\bar{\lambda}^{1} (i,j)$& $\bar{\lambda}^{0.5} (i,j)$ & $\widehat{\lambda^m_{\mathbb{V}}}(i,j)$  \\ \cline{2-5} \cline{6-9} 
$i=1$& 0.5 & 0.5  & 0.5   & 0.49 & & & &     \\
$i=2$& 1.5 & 1    & 0.75  & 0.78  & 1.5& 1.5  & 1.5  & 1.45  \\
$i=3$& 2.5 & 1.5  & 1      & 1.29  & 2.5& 2     & 1.75  & 1.84 \\
$i=4$& 3.5 & 2    & 1.25  & 1.66 & 3.5 & 2.5  & 2      & 2.20  \\
\hline
\end{tabular}

\begin{tabular}{@{\extracolsep{4pt}}lrrrrrrrr@{}} \hline
 & \multicolumn{4}{c}{$j=3$} & \multicolumn{4}{c}{$j=4$}  \\ 
 & $d/2$ & $\bar{\lambda}^1 (i,j)$ & $\bar{\lambda}^{0.5} (i,j)$ & $\widehat{\lambda^m_{\mathbb{V}}}(i,j)$ & $d/2$ & $\bar{\lambda}^{1} (i,j)$& $\bar{\lambda}^{0.5} (i,j)$ & $\widehat{\lambda^m_{\mathbb{V}}}(i,j)$  \\ \cline{2-5} \cline{6-9} 
$i=1$ & & & & & & & & \\
$i=2$ & & & & & & & & \\
$i=3$ &2.5 & 2.5 & 2.5& 2.49 & & & & \\
$i=4$ &3.5 & 3& 2.75 & 2.79 & 3.5 & 3.5 & 3.5 & 3.52 \\
\hline
\end{tabular}

\label{tab_bmm}
\end{center}
\end{table*}

Next, we compare the performances of BIC, $\overline{\rm sBIC}_{0.5}$, $\overline{\rm sBIC}_{1}$, WBIC, and WsBIC by selecting the number of mixture components. We set the true number of components to two and the number of candidate components from one to four. We conduct 200 simulations for each $n=10, 20, 50$ to select the number of components by sBIC, WsBIC, WBIC and BIC. We compute sBIC, WsBIC and BIC using the R package {\it sBIC} \citep{Weihs2016} and WBIC using R package {\it RStan} \citep{RStan}.

\begin{figure*}
\includegraphics[width=\textwidth]{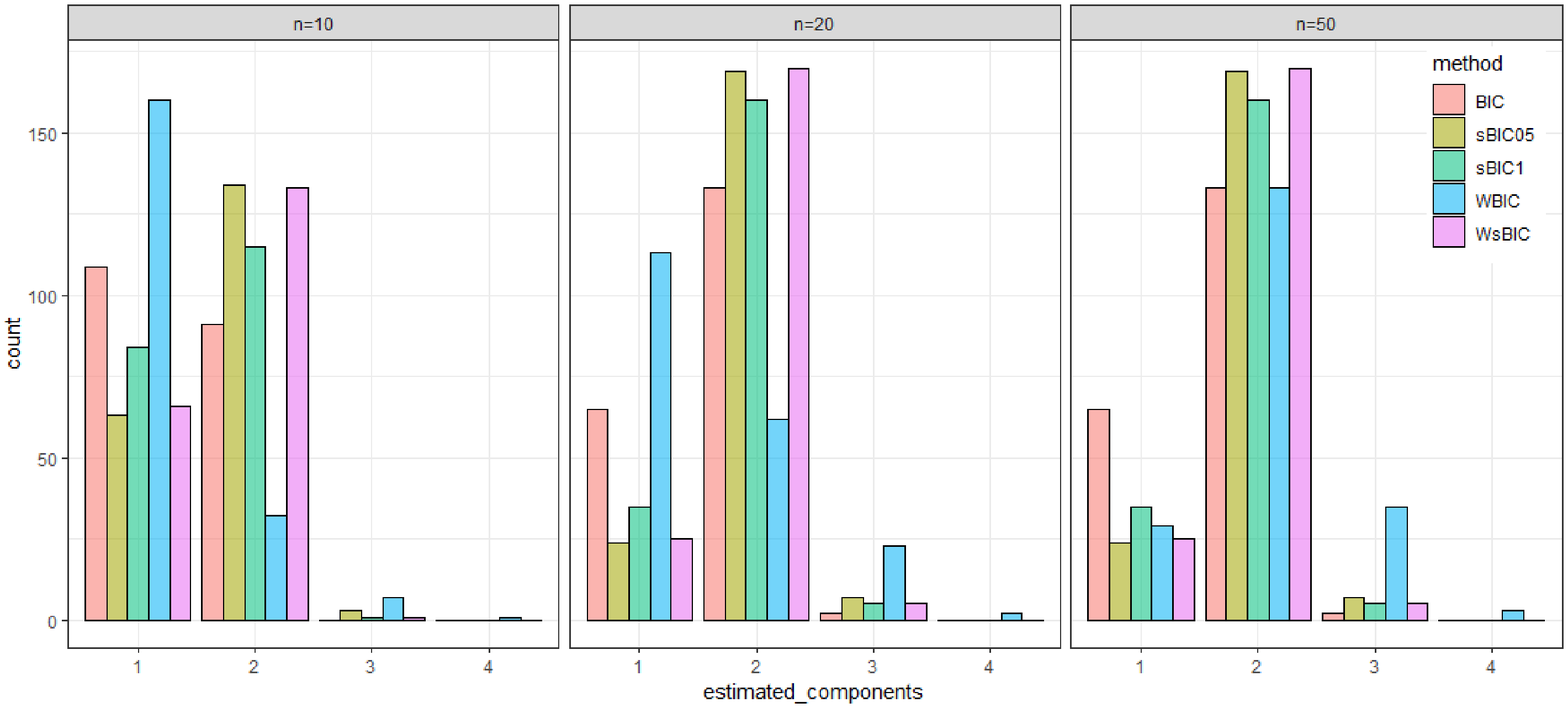}
\caption{Performances of BIC, $\overline{\rm sBIC}_{0.5}$, $\overline{\rm sBIC}_{1}$, WBIC, and WsBIC in the binomial mixture model}
\label{fig_bmm}
\end{figure*}

The results are shown in Figure \ref{fig_bmm}. 
First, the performances of $\overline{\rm sBIC}_{0.5}$ and WsBIC are best, and that of $\overline{\rm sBIC}_{1}$ is slightly inferior. For small sample sizes, WBIC does not perform well. 
Note that the binomial mixture model has tight upper bounds $\bar{\lambda}^{0.5}(i,j)$, which is shown by the complicated analysis, but few tight upper bounds are known. Therefore, we may not expect the performance of $\overline{\rm sBIC}_{0.5}$ in other singular models.
However, we may expect the performance of $\overline{\rm sBIC}_{1}$ because $\bar{\lambda}^{1}(i,j)$ is derived from simple parameter counting.
In addition, since we can think of BIC as sBIC with the trivial upper bound $d/2$, when we are able to use only loose upper bounds of the real log canonical thresholds, the performance of sBIC is expected to be much worse than that of WsBIC.

\subsection{Application of WsBIC}

In this section, we apply WsBIC to cormorant census data \citep{McCrea2014}. The data were collected from the Vors$\o$ colony in 1994. We use the April successful breeder census data to determine the number of classes of cormorants. The data are summarized in Table \ref{tab_corm}. f$_t$ represents the number of individuals captured $t$ times in 30 visits.

We compare the binomial mixture models with one to four components. We use $\widehat{\lambda^m_\mathbb{V}}(i,j)$ in Section \ref{sec:bmm_exp} for WsBIC, implemented in the R package {\it sBIC} \citep{Weihs2016}. We set the prior of $\pi_h$ to the flat Dirichlet distribution and $p_h \sim {\rm logit}^{-1}({\rm Unif}(-\infty, \infty))$, and use the Hamiltonian Monte Carlo method to obtain the posteriors to compute WBIC using the R package {\it RStan} \citep{RStan}.

The posterior model probabilities calculated using BIC, WBIC and WsBIC are shown in Figure \ref{fig_corm}. The model with the highest posterior model probability using BIC and WBIC has two components, whereas that using WsBIC has three components.
As seen in the simulation study, BIC and WBIC tend to select smaller models in cases of small sample sizes, and this might be reflected in the result.

For the mixture model with three components, we obtain $(\widehat{\pi_1}, \widehat{\pi_2}, \widehat{\pi_3}) = (0.438, 0.507, 0.055)$ and $(\widehat{p_1},\widehat{p_2},\widehat{p_3}) = (0.095, 0.302, 0.459)$. Since the weight of the third group is small, BIC and WBIC might not detect the third group.

\begin{table*}
\caption{Cormorant census data at the Vors$\o$ colony in April, 1994. f$_t$ represents the number of individuals captured $t$ times.}
\begin{tabular}{lllllllllllllllllllll} \hline
f$_1$ & f$_2$ & f$_3$ & f$_4$ & f$_5$ & f$_6$ & f$_7$ & f$_8$ & f$_9$ & f$_{10}$ & f$_{11}$ & f$_{12}$ & f$_{13}$ & f$_{14}$ & f$_{15}$ & f$_{16}$ & f$_{17}$ & f$_{18}$ & f$_{19}$ & f$_{20}$ & f$_{21}$\\ \hline
13 & 14 & 10 & 8 & 11 & 7 & 7 & 12 & 7 & 9 & 6 & 10 & 7 & 2 & 0 & 3 & 1 & 0 & 0 & 0 & 1\\ \hline
\end{tabular}
\label{tab_corm}
\end{table*}

\begin{figure}
\begin{center}
\includegraphics[width=0.5 \textwidth]{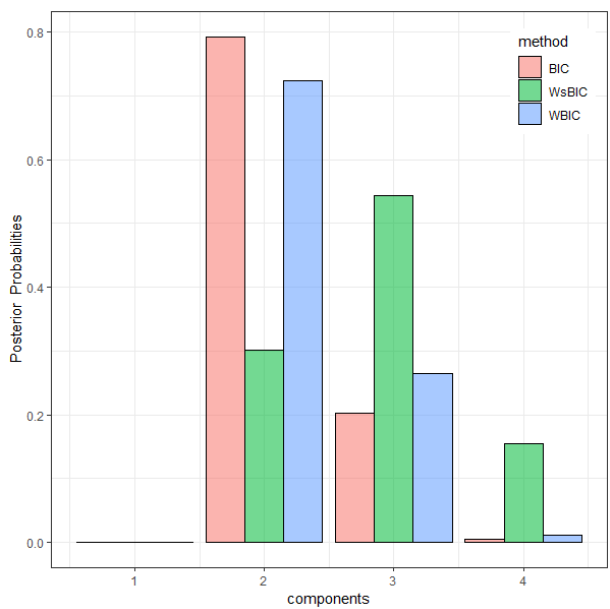}
\caption{Posterior model probabilities using BIC, WBIC, and WsBIC for cormorant census data}
\label{fig_corm}
\end{center}
\end{figure}

\section{Discussion}
\label{sec:dis}
In this paper, we propose a new consistent estimator $\widehat{\lambda^1_{\mathbb{V}}}$ of the real log canonical threshold and its multiple version  $\widehat{\lambda^m_{\mathbb{V}}}$.
In contrast to the existing method $\widehat{\lambda_{\mathbb{E}}}$, we show that the proposed method $\widehat{\lambda^m_{\mathbb{V}}}$ has an optimal hyperparameter $\widetilde{c^*}$.
In addition, we show that $\widehat{\lambda^m_{\mathbb{V}}}$ is consistent.
In our simulation studies, $\widehat{\lambda^m_{\mathbb{V}}}$ has the lowest bias of all methods and we can control the variance for $\widehat{\lambda^m_{\mathbb{V}}}$ by increasing $m$.

For applications using estimated real log canonical thresholds, we propose WsBIC. 
WsBIC does not require the theoretical values of real log canonical thresholds, whereas sBIC does. Hence, WsBIC is widely applicable.
Simulation studies show that WsBIC performs competitively compared to sBIC and $\overline{\rm sBIC}$ with tight upper bounds.

In future works, this research can be used to evaluate MCMC in singular models, since as \cite{Watanabe2018} states, \textit{``In order to check how accurate MCMC approximates the posterior distribution in singular cases, the real log canonical threshold would be a good index for a given set of a true distribution, a statistical model, and a prior.''}
Another future topic is to estimate the multiplicities $\mathfrak{m}$ of real log canonical thresholds, which would give a better evaluation of the marginal likelihood via WsBIC. Finally, identifying an optimal temperature $t^*$ is important to improve WBIC and also to construct an unbiased estimator of real log canonical thresholds.

\begin{appendices}
\section{Assumptions}

In this appendix, we introduce the technical four assumptions that the singular learning theory requires \citep{Watanabe2000, Watanabe2001a, Watanabe2009b}.

\noindent {\it Assumption (a)}. The set of parameters ${\bf \Omega}$ is a compact set in $\mathbb{R}^d$ and can be defined by analytic functions $\pi_1, \pi_2, ..., \pi_k$;
\begin{equation*}
{\bf \Omega} = \{ \theta \in \mathbb{R}^d | \pi_1(\theta)\ge 0,..., \pi_k(\theta)\ge 0 \}.
\end{equation*}

\noindent {\it Assumption (b)}. The prior distribution $\varphi(\theta)$ can be decomposed as the product of a non-negative analytic function $\varphi_1$ and a positive differentiable function $\varphi_2$;
\begin{equation*}
\varphi(\theta) = \varphi_1(\theta) \varphi_2(\theta).
\end{equation*}

\noindent {\it Assumption (c)}. Let $s \ge 6$ and
\begin{equation*}
L^s(q) = \left\{ f(x) | \Bigl(\int |f(x)|^s q(x)dx \Bigr)^{1/s} < \infty \right\}
\end{equation*}
be a Banach space. There exists an open set ${\bf \Omega}' \supset {\bf \Omega}$ such that for $\theta \in {\bf \Omega}'$ the map $\theta \mapsto \log q(x)/p(x|\theta, M)$ is an $L^s(q)$-valued analytic function.

\noindent {\it Assumption (d)}. Let ${\bf \Omega}_\epsilon$ be the set
\begin{equation*}
{\bf \Omega}_\epsilon = \{ \theta \in {\bf \Omega} | K(\theta) \leq \epsilon \},
\end{equation*}
where $K(\theta) = \int q(x) \log q(x)/p(x|\theta, M) dx$. There exists a pair of positive constants $(\epsilon, c)$ such that
\begin{equation*}
 \mathbb{E}_X \{ \log q(X)/p(X|\theta, M) \}  \ge c \mathbb{E}_X [ \{ \log q(X)/p(X|\theta, M) \}^2], \quad \forall \theta \in {\bf \Omega}_\epsilon.
\end{equation*}

\end{appendices}

\bibliographystyle{chicago}

\bibliography{rlct.bib}
\end{document}